\documentclass[english]{llncs}
\usepackage[T1]{fontenc}
\usepackage[latin9]{inputenc}
\usepackage{tikz}
\usepackage{verbatim}
\usepackage{amsmath}
\usepackage{amssymb}
\usepackage[mathcal]{euscript}
\usepackage{xcolor}

\makeatletter
\usepackage[draft]{botex}

\makeatother

\usepackage{babel}
\begin{document}
\global\long\def\rema{\mathrm{rem}}
\global\long\def\eps{\epsilon}
\global\long\def\opt{\ensuremath{\mathrm{OPT}}}
 \global\long\def\poly{\ensuremath{\mathrm{poly}}}
 \global\long\def\idle{\ensuremath{{idle}}}

\newcommand{\types}{\ensuremath{{T}}}
\newcommand{\smalljobs}{\ensuremath{J_\mathrm{s}}}
\newcommand{\bigjobs}{\ensuremath{J_\mathrm{b}}}

\newcommand{\topmach}{\ensuremath{M_{\mathrm{vh},\ell}}}
\newcommand{\hugemach}{\ensuremath{M_{\mathrm{h},\ell}}}
\newcommand{\smallmach}{\ensuremath{M_{\mathrm{s},\ell}}}
\newcommand{\allsmallmach}{\ensuremath{M_{\mathrm{s}}}}

\defineclasses 

\title{Scheduling Unrelated Machines \\
of Few Different Types}
\author{Vincenzo Bonifaci\inst{1} \and Andreas Wiese\inst{2}} \institute{ IASI-CNR, Rome, Italy \\ \email{vincenzo.bonifaci@iasi.cnr.it} \and Sapienza University of Rome, Italy \\ \email{wiese@dis.uniroma1.it}}
\maketitle

\begin{abstract}
A very well-known machine model in scheduling allows the machines to be \emph{unrelated}, modelling jobs that might have different characteristics on each machine. Due to its generality, many optimization problems of this form are very difficult to tackle and typically \ccapx-hard. However, in many applications the number of different \emph{types} of machines, such as processor cores, GPUs, etc.~is very limited. In this paper, we address this point and study the assignment of jobs to unrelated machines in the case that each machine belongs to one of a fixed number of types and the machines of each type are identical. We present polynomial time approximation schemes (PTASs) for minimizing the makespan for multidimensional jobs with a fixed number of dimensions and for minimizing the $L_{p}$-norm. In particular, our results subsume and generalize the existing PTASs for a constant number of unrelated machines and for an arbitrary number of identical machines for these problems. We employ a number of techniques which go beyond the previously known results, including a new counting argument and a method for making the concept of sparse extreme point solutions usable for a convex program. 
\end{abstract}

\section{Introduction}
One of the most general models in machine scheduling is the model of \emph{unrelated machines}, where the characteristics of each job depend on the machine that executes it. As the term ``unrelated'' suggests, these characteristics might be completely different for each machine. This is a very general model which causes a significant increase in necessary algorithmic effort and complexity in comparison with simpler models. 
For instance, for makespan minimization there are polynomial time approximation schemes known for an arbitrary number of identical and uniform machines~\cite{Hochbaum:1987,Hochbaum:1988}, but for unrelated machines approximating the makespan within a ratio lower than $3/2$ is \ccnp-hard. 
Even more, improving upon the best known 2-approximation algorithm \cite{Lenstra:1990} has been an important open problem in scheduling for more than 20 years. 

In recent years, the design of modern hardware architectures has seen an advent of heterogeneous processors in a system, e.g., the cores of a CPU, graphics processing units or floating point units. These devices usually have very different characteristics since they are especially designed for certain operations. However, the total number of different \emph{types} of processors in a system is usually very limited. This motivates the study of unrelated machine scheduling in the setting that the given machines (arbitrarily many) are partitioned into a constant number of types. 
This setting subsumes and generalizes two classical scenarios: scheduling an arbitrary number of identical parallel machines, and scheduling a fixed number of unrelated parallel machines. We study two problems in this setting: first, the assignment of multidimensional jobs (in an application, each dimension may correspond to a scarce resource such as execution time, memory requirement, etc.). The objective is to minimize the makespan across all dimensions. 
For any $\eps>0$ and any fixed number of dimensions, we provide a $(1+\eps)$-approximation algorithm for this problem. Our second problem is the minimization of the $L_{p}$ norm of the load vector of one-dimensional jobs (instead of the $L_{\infty}$ norm). In this case we derive a PTAS for any fixed $p>1$.

\paragraph{Related work.}
Lenstra, Shmoys and Tardos~\cite{Lenstra:1990} and Shmoys and Tardos \cite{Shmoys:1993} presented a 2-approximation algorithm for makespan minimization of an arbitrary number of unrelated machines. It is known that approximating the same problem within a ratio lower than $3/2$ is \ccnp-hard \cite{Lenstra:1990}. In fact, the problem remains \ccapx-hard for any fixed $L_{p}$ norm \cite{Azar:2004}. 

Azar and Epstein \cite{Azar:2005} considered the minimization of the $L_{p}$ norm when scheduling unrelated parallel machines. They give a 2-approximation algorithm for any fixed $L_{p}$ norm ($p>1$) and a PTAS for the case of a fixed number of machines. They also give a $(D+1)$-approximation algorithm for the minimization of the $L_{p}$ norm of the generalized load vector of $D$-dimensional jobs, for any fixed $p>1$. 
The approximation ratio of the former algorithm was later improved to less than 2 by Kumar et al.~\cite{Kumar:2009}. 

For the special case of identical machines, a PTAS for makespan minimization was given by Hochbaum and Shmoys \cite{Hochbaum:1987}. 
A PTAS was later given for general $L_{p}$ norms by Alon et al.~\cite{Alon:1998}.
For the $D$-dimensional makespan minimization problem with identical machines, Chekuri and Khanna \cite{Chekuri:2004} provide a $O(\log^{2}D)$-approximation when $D$ is arbitrary, and a PTAS when $D$ is fixed.
A PTAS for makespan minimization on uniform machines has been provided by Hochbaum and Shmoys \cite{Hochbaum:1988}. 

\paragraph{Our contribution.}
In this paper we study scheduling problems on unrelated machines which are partitioned into a constant number of types, such that two machines of the same type are identical. 
In this setting, we present polynomial time approximation schemes for the problems of minimizing the makespan of $D$-dimensional jobs, for constant $D$, and for minimizing the $L_p$-norm of 1-dimensional jobs, for constant $p > 1$. 
Both results subsume and generalize the known PTASs for an arbitrary number of identical machines \cite{Alon:1998,Chekuri:2004,Hochbaum:1987} and for a constant number of unrelated machines \cite{Azar:2005,Horowitz:1976,Lenstra:1990}. Not surprisingly, certain ideas in the latter algorithms are useful in our setting as well, for example, geometric rounding and enumeration techniques.
However, obtaining our results requires non-trivial extensions of the known methods since we face obstacles that do not occur in either of the two subsumed cases. In particular, there is no direct way to enumerate the assignment of the large jobs, i.e., jobs that are longer than an $\epsilon$-fraction of the load of their respective machine in an optimal solution. 
The reason is that a job could be large on machines of one type and small on the other and in total there can be a superconstant number of large jobs.

To remedy this, in our algorithm for makespan minimization we use a linear program to assign jobs to slots which we enumerate for large jobs. For rounding this (sparse) linear program, we use an iterative rounding approach \cite{Jain:2001,Lau:2011}. We identify constraints in the LP that can be dropped without affecting the computed solution too much. To this end, we introduce a new counting argument for the number of non-integral variables, which is crucial to the approach and may be more generally applicable. This counting argument is one of the novelties of our contribution. It was also successfully used in~\cite{MarchettiRutten2012}, which evolved in parallel to the research presented here.

When minimizing the $L_p$-norm, matters are even more complicated. Unlike in the PTAS for an arbitrary number of identical machines, we cannot assume that all machines will have roughly the same load (on identical machines, this would hold after some preprocessing~\cite{Alon:1998}). However, this property is important when classifying jobs into ``large'' and ``small''.
We overcome this obstacle by identifying some properties of the optimal solution that we can enumerate in polynomial time. Moreover, due to the convexity of the objective function, we cannot use a linear program but have to employ a convex program~(CP). Unfortunately, we cannot assume that the computed CP solution has the same sparseness properties of an LP extreme point solution, which is what an iterative rounding strategy needs. To address this, we use a trick: we take the computed CP solution and formulate a linear program based on it, which we show has to be feasible. We then compute a (sparse) extreme point solution of this LP and perform an enhanced version of the iterative rounding algorithm for makespan minimization. Hence, we show how to make the concept of sparse extreme point solutions usable for a convex program, an approach that might be useful in other settings as well.


\section{Problem Definitions}
For both of the problems studied in this article, the task is to assign a set of jobs $J$ to a set of unrelated parallel machines $M$. We let $n$ and $m$ denote the number of jobs and machines, respectively. 

\paragraph{Makespan minimization of $D$-dimensional jobs.} 
The input is represented by a positive integer $c_{i,j}^{d}$ for each job $j \in J$, each machine $i \in M$, and each dimension $d\in\{1,...,D\}$. The objective is to minimize the makespan, given by $$\max_{d\in\{1,...,D\}}\max_{i \in M}\sum_{j\in J_{i}}c_{i,j}^{d},$$ where for each machine $i$ the set $J_{i}$ denotes the jobs assigned to $i$. In other words, each machine $i$ has a load of $\sum_{j\in J_{i}}c_{i,j}^{d}$ in dimension $d$, and the objective is to minimize the maximum load of all machines across all dimensions. 

We further assume that the machines are partitioned into a set $\types=\{\ell_1,\ldots,\ell_K\}$ of $K$ distinct types, where $K$ is assumed to be some constant. For two machines $i,i'$ of the same type, one has $c_{i,j}^{d}=c_{i',j}^{d}$ for each job $j$ and dimension $d$. 

\paragraph{$L_{p}$-norm minimization of one-dimensional jobs.} 
The input is represented by a positive integer $c_{i,j}$ for each job $j \in J$ and each machine $i \in M$. The objective is to minimize 
$$\left\Vert \left(\sum_{j\in J_{1}}c_{1,j},...,\sum_{j\in J_{m}}c_{m,j}\right)\right\Vert _{p}=\left(\sum_{i=1}^{m}\left(\sum_{j\in J_{i}}c_{i,j}\right)^{p}\right)^{1/p}.$$
We again assume that the machines are partitioned into a set $\types=\{\ell_1,\ldots,\ell_K\}$ of $K$ distinct types, for some constant $K$. For two machines $i,i'$ of the same type, one has $c_{i,j}=c_{i',j}$ for each job $j$.

\section{Makespan Minimization of Multidimensional Jobs}
\label{sec:Makespan-minimization}

We present a polynomial time $(1+\epsilon)$-approximation algorithm
for makespan minimization in $D$ dimensions on unrelated machines of
at most $K$ types, with $D$ and $K$ being constants.

Let $\eps>0$ and suppose that we are given an instance of our problem.
First, we establish a binary search framework to estimate the optimal
makespan. Hence, by suitable scaling, it remains to give an algorithm
which either asserts that there is no solution with makespan at most
1, or which computes a job assignment with makespan
at most $1+\epsilon$. The
general idea of this algorithm is as follows: For each machine, we
classify job into \emph{large }and \emph{small }jobs. With an enumeration
procedure, we enumerate patterns for the big jobs on the machines.
One of the enumerated patterns will correspond to an optimal solution.
Having guessed the correct pattern, the remaining problem is to assign
each job either to a slot in the pattern (then the job is big on its
machine) or to the remaining space of the machine (then it is small
on its machine). We model this problem as a linear program. Given
that the LP is feasible (otherwise we know that the enumerated pattern
was wrong), using an iterative rounding approach, we compute a solution
with makespan $1+O(D\cdot\eps)$. 

Now we present our algorithm in detail. We call a job $j$ \emph{large
on machine} $i$ if there is a dimension $d$ such that $c_{i,j}^{d}\ge\eps$.
Note that in each feasible solution, the number of large jobs on each
machine is bounded by $\floor{D/\eps}$ 
which is a constant. For technical reasons, for each large job $j$
on a machine $i$ we redefine each value $c_{i,j}^{d}$ by setting
it to $\max\{c_{i,j}^{d},\eps^{2}/D\}$. This does not increase the
makespan of any feasible solution by more than $\epsilon$, as the
following proposition shows.
\begin{proposition}
\label{prop:largify}
Let $J_i$ be a set of jobs on a machine $i$ such that $\sum_{j\in J_i}c_{i,j}^{d}\le1$
for each dimension $d$. Let $\smalljobs \subseteq J_i$ and $\bigjobs \subseteq J_{i}$ denote the small and large jobs in $J_i$, respectively. Then $\sum_{j \in \smalljobs}c_{i,j}^{d}+\sum_{j\in \bigjobs}\max\{c_{i,j}^{d},\eps^{2}/D\}\le1+\epsilon$
for each dimension $d$.
\end{proposition}
Next, we round up each input value $c_{i,j}^{d}$ to the next greater
power of $\frac{1}{1+\epsilon}$. This does not increase the objective
by more than a factor $1+\eps$. After this preparation, we enumerate
the patterns of the big jobs on each machine. Intuitively, a pattern
for a machine $i$ describes the sizes of the jobs running on $i$.
We call a vector $q=(q^{1},...,q^{D})$ a \emph{large job type }if
each $q^{d}$ is a power of $1+\epsilon$ and $\epsilon^{2}/{D}\le q^{d}\le1$.
Let $Q$ be the set of all large job types. Note that since $D$ and
$\epsilon$ are constants, $|Q|$ is also bounded by a constant. We
call a vector $\pi\in\{0,\ldots,\floor{D/\eps}\}^{Q}$
a \emph{pattern of large jobs }for a machine. For each machine type
$\ell$, we enumerate how many machines follow which pattern. With
$\kappa$ being the (constant) number of possible patterns for a machine,
there are at most $(m+1)^{K\cdot\kappa}\in\poly(m)$ combinations
for the patterns of all machines. Note that since all machines of
the same type are identical, the actual ordering of the machines of
the same type does not matter.

Assume for the ease of presentation that we correctly guessed the pattern which
corresponds to an optimal solution. For a machine $i$ with a pattern
$\pi=(\pi_{q})_{q\in Q}$, we obtain $\pi_{q}$ \emph{slots} for large
jobs of type $q$ and a certain amount of remaining capacity $\rema^{d}(i)$
in each dimension $d$. Denote by $S$ the set of all slots. It remains
to determine an assignment of the jobs to the slots and to the remaining
capacity on each machine. Of course, a job $j$ can only be assigned
to a slot $s$ on a machine $i$ if its size on $i$ corresponds to
$s$. Also, a job $j$ can only be assigned to the remaining space
on $i$ if $j$ is small on $i$ (otherwise we would have enumerated
a slot for it). We model this assignment problem with the following
linear program, denoted by Slot-LP:

\begin{align}
\mathrm{\textrm{(Slot-LP)}} &  & \sum_{i\in M}x_{i,j}+\sum_{s\in S}x_{s,j} & =1 &  & \forall j\in J\label{eq:total_assigned}\\
 &  & \sum_{j\in J}x_{s,j} & \le1 &  & \forall s\in S\label{eq:slot-constraint}\\
 &  & \sum_{j\in J}c_{i,j}^{d}\cdot x_{i,j} & \leq\rema^{d}(i) &  & \forall i\in M,\,\forall d=1,\ldots,D\label{eq:idle-constraint}\\
 &  & x_{i,j} & \ge0 &  & \forall i\in M,\,\forall j\in J\nonumber \\
 &  & x_{s,j} & \ge0 &  & \forall s\in S,\,\forall j\in J.\nonumber 
\end{align}
If Slot-LP is infeasible, then in particular there is no integral
solution and the enumerated pattern was wrong. Now assume that Slot-LP
is feasible. With an iterative rounding approach (similar to \cite{Jain:2001,Lau:2011})
we round the fractional solution. We define $LP_{0}$ to be the Slot-LP.
In each iteration $t$ we solve a linear program $LP_{t}$ which has
the same structure as the Slot-LP, but it will involve only a reduced
set of machines and jobs. Consider an iteration $t$. We compute an
extreme point solution $x^{*}$ to $LP_{t}$. We say that a job $j$
is \emph{fractionally assigned to a machine $i$} if $x_{i,j}^{*}\in(0,1)$
and it is \emph{fractionally assigned to a slot $s$} if $x_{s,j}^{*}\in(0,1)$.
Using the sparsity of extreme point solutions together with a useful
counting argument, we derive the following lemma.
\begin{lemma}
\label{lem:extreme-point-solutions}In $x^{*}$ there is either a
machine $i$ which has at most $2D$ small jobs fractionally assigned
to it, or a slot $s$ which has at most 2 jobs fractionally assigned
to it.\end{lemma}
\begin{proof}
Denote by $I$ the number of variables which equal 1 and by $F$ the
number of fractional variables in $x^{*}$. Let $s'$ be the number
of slot constraints which are still in $LP_{t}$ and let $m'$ be
the number of machines for which there are constraints of type \eqref{eq:idle-constraint}
in $LP_{t}$. Since $x^{*}$ is an extreme point solution, we have
that $I+F\le n+s'+D\cdot m'$. Also, it holds that $n\le I+F/2$.
We claim that $F\le2\cdot s'+2m'\cdot D$. Assume on the contrary
that $F>2\cdot s'+2m'\cdot D$. But then 
\[
n\le I+F/2=(I+F)-(F/2)<(n+s'+D\cdot m')-(s'+m'\cdot D)=n
\]
 which is a contradiction.
For proving the main claim of the lemma, if each machine $i$ had
strictly more than $2D$ jobs fractionally assigned to it and each
slot had more than 2 jobs fractionally assigned to it, then $F>2\cdot s'+2m'\cdot D$
which is a contradiction. \qed 
\end{proof}

First, we fix all variables that have integral values in $x^{*}$.
If there is a machine~$i$ such that in $x^{*}$ there are at most
$2D$ small jobs fractionally assigned to it, we remove all constraints
of type~\eqref{eq:idle-constraint} for machine $i$ from the LP.
This is justified since after fixing the integral variables of $x^{*}$,
any solution for the remaining variables can violate the constraint~\eqref{eq:idle-constraint}
for $i$ by at most an additional value of $2D\cdot\eps$. The second
case of Lemma~\ref{lem:extreme-point-solutions} is that there is
a slot $s$ which has at most two jobs $j_{1},j_{2}$ fractionally
assigned to it. Intuitively, we seek an integral solution in which
either $j_{1}$ or $j_{2}$ will be assigned to $s$. We model this
by removing $j_{1}$, $j_{2}$, and $s$ from the instance and by
adding a new \emph{artificial job }$j_{0}$ with the following characteristics: 
\begin{itemize}
\item $j_{0}$ is allowed to be assigned to any slot where either $j_{1}$
or $j_{2}$ were allowed to be assigned to, 
\item for each machine $i$ on which only job $j_{1}$ but not job $j_{2}$
is small, then $c_{i,j_{2}}^{d}:=c_{i,j}^{d}$ for all dimensions
$d$, 
\item similarly, for each machine $i$ on which only job $j_{2}$ but not
job $j_{1}$ is small, then $c_{i,j_{0}}^{d}:=c_{i,j_{2}}^{d}$ for
all dimensions $d$, 
\item for each machine $i$ on which both $j_{1}$ and $j_{2}$ are small,
we define $c_{i',j_{0}}^{d}:=(x_{i,j_{1}}^{*}/(x_{i,j_{1}}^{*}+x_{i,j_{2}}^{*}))\cdot c_{i,j_{1}}^{d}+(x_{i,j_{2}}^{*}/(x_{i,j_{1}}^{*}+x_{i,j_{2}}^{*}))\cdot c_{i,j_{2}}^{d}$
for all dimensions $d$.
\end{itemize}
We say that $j_{0}$ \emph{subsumes} the jobs $j_{1}$ and $j_{2}$
and that $j_{0}$ \emph{disposes} the slot $s$. If $j_{1}$ or $j_{2}$
were already artificial, we say that $j_{0}$ also subsumes all jobs
which were subsumed by $j_{1}$ or $j_{2}$, and similarly disposes
all slots which were disposed by $j_{1}$ or $j_{2}$. Denote by $LP_{t+1}$
the resulting linear program. Propositions~\ref{pro:remove-machine} and \ref{pro:remove-slot} in the appendix
show that $LP_{t+1}$ is feasible (in fact, a solution can be constructed
from $x^{*}$). Note that by Lemma~\ref{lem:extreme-point-solutions}
we perform at least one of the two procedures above and hence, in $LP_{t+1}$
either the number of machines or the number of jobs is strictly less
than in $LP_{t}$. Hence, after at most $n+m$ iterations, we obtain
an integral vector $\bar{x}$ which assigns all artificial jobs and
all original jobs which are not subsumed by an artificial job. 
\begin{lemma}
In the vector $\bar{x}$, each slot has at most one job assigned to
it and $\sum_{j\in J}c_{i,j}^{d}\cdot\bar{x}_{i,j}\leq \rema^{d}(i)+2D\cdot\epsilon$
for each machine $i$ and each dimension $d$. 
\end{lemma}
It remains to transform the vector $\bar{x}$ to a solution for all
real (i.e., non-artificial) jobs, rather than for some of the real jobs and the artificial jobs. One can show that
for each artificial job $j$ there is a set $J_{j}$ of real jobs which were subsumed by $j$ and a set $S_{j}$
of slots which were disposed by $j$. In particular, we can show that
for each job $j'\in J_{j}$ there is a feasible assignment of the
jobs in $J_{j}\setminus\{j'\}$ to the slots $S_{j}$. 
The following lemma proves formally the important properties of the artificial jobs.
\begin{lemma}
\label{lem:tree} For each artificial job $j$ there is
a set of original jobs $J_{j}$ and a set of slots $S_{j}$ such that 
\begin{enumerate}
\item $J_{j}\cap J_{j'}=\emptyset$ and $S_{j}\cap S_{j'}=\emptyset$ for
any two artificial jobs $j,j'$, $j\ne j'$; 
\item $c_{i,j}^{d}$ is a convex combination of $\{c_{i,j'}^{d}|j'\in J_{j}\}$
for each machine $i$ and each dimension $d$, 
\item no slot in $S_{j}$ is used by $\bar{x}$, 
\item no job in $J_{j}$ is assigned by $\bar{x}$, and 
\item for each $j'\in J_{j}$ there is a feasible assignment of the jobs
in $J_{j}\setminus\{j'\}$ to the set of slots $S_{j}$. 
\end{enumerate}
\end{lemma}

Using the above properties of the artificial jobs we transform $\bar{x}$ to an
integral solution for all real jobs. Let $j$ be an artificial job assigned to
some slot $s\notin S_{j}$. Then there is a job $j'\in J_{j}$ that
can be assigned to $s$ and we assign all jobs in $J_{j}\setminus\{j'\}$
to the slots $S_{j}$.
For replacing the artificial jobs which are not assigned to a slot
but to the remaining space on a machine we consider all those jobs
on a machine at the same time. 
\begin{lemma}
\label{lem:replace-artificial-jobs}
Let $i$ be a machine and let $RJ_{i}$ and $AJ_{i}$ denote all real and artificial jobs, respectively, which were assigned to the remaining space on $i$. There is an integral assignment $x'$ of the jobs in $\cup_{j\in AJ_{i}}J_{j}$ to $i$ and the slots in $\cup_{j\in AJ_{i}}S_{j}$ such that each slot gets at most one job assigned to it and $$\sum_{j\in RJ_{i}}c_{i,j}^{d}\cdot\bar{x}_{i,j}+\sum_{j\in\cup_{j'\in AJ_{i}}J_{j'}}c_{i,j}^{d}\cdot x'_{i,j}\leq \rema^{d}(i)+3D\cdot\epsilon.$$ 
\end{lemma}
\begin{proof}
A solution to the following linear program can be extracted from $\bar{x}$:
\begin{align*}
(\textrm{Art-LP})_i \qquad \sum_{j\in J_{j'}}x_{i,j} & \ge 1 &  \qquad\forall j'\in AJ_{i}\\
 \sum_{j\in RJ_{i}}c_{i,j}^{d}\cdot \bar{x}_{i,j}+\sum_{j\in\cup_{j'\in AJ_{i}}J_{j'}}c_{i,j}^{d}\cdot x_{i,j} & \leq \rema^{d}(i)+2D\cdot\epsilon & \qquad\forall d \in \{1,\ldots,D\} \\
 x_{i,j} & \ge 0 &  \qquad \forall j\in\cup_{j'\in AJ_{i}}J_{j'}. 
\end{align*}
In an extreme point solution to $\textrm{(Art-LP)}_{i}$ there can
be at most $|AJ_i|+D$ non-zero entries. Hence, by a standard counting argument
one can show that there are at most 
$D$ sets $J_{j'}$ from which a job is fractionally assigned.
Rounding up all fractional values and assigning the other jobs in
the sets $J_{j}$ to the respective slots in $S_{j}$ yields a solution
with the claimed properties.\qed
\end{proof}
Applying the procedure of Lemma~\ref{lem:replace-artificial-jobs}
to each machine $i \in M$ yields an integral assignment of jobs with a
makespan of $1+3 D\eps$ in each dimension. Together with the
binary search framework, this yields our main theorem of this section.
\begin{theorem}
Let $D, K \in \Nat$ be constants. For any $\eps>0$ there is a $(1+\epsilon)$-approximation
algorithm for makespan minimization in $D$ dimensions on unrelated machines
of at most $K$ types.
\end{theorem}

\section{$L_{p}$-norm Minimization}

In this section, we present a $(1+\epsilon)$-approximation algorithm
for assigning jobs on machines with $K$ types to minimize the $L_{p}$-norm
of the loads of the machines, for $1 < p < \infty$. Note that since we work with the $L_{p}$-norm,
we assume the jobs to be 1-dimensional.

Our strategy is the following: First, we reduce the complexity of
the problem by enumerating certain structural properties of the optimal
solution. Those will include the patterns for the big jobs (like in
Section~\ref{sec:Makespan-minimization}) but also certain other
information which are only important when working with $L_{p}$-norms.
For the remaining problem we will formulate and solve a convex programming relaxation. 
Unlike for
linear programs, we cannot assume the obtained CP solution to be sparse
(like e.g., an extreme point solution). However, from the obtained
CP-solution we will derive a feasible solution to a linear program.
From this LP we obtain an extreme point solution which we can round
by enhancing the iterative rounding scheme presented in Section~\ref{sec:Makespan-minimization}.
Hence, in our algorithm we make the concept of (sparse) extreme point
solutions of LPs usable for a convex program.

Let $\epsilon>0$ and $p>1$ be constants. First, instead
of minimizing $\left\Vert g\right\Vert _{p}$, where $g$ denotes
the vector given by the loads of the machines, we minimize~
$\norm[p]{g}^{p}$.
Note that a $(1+\eps)^p$-approximation algorithm for the latter translates to
a $(1+\eps)$-approximation for the former. 
Suppose that we are given an instance of our problem. We start by
enumerating certain properties of the optimal solution. In an optimal
solution, there might be some machines which execute only one job.
Intuitively, these jobs are quite large. We call those 
machines \emph{huge}. For each type~$\ell$, we enumerate the number
of huge machines, denoted by~$h_{\ell}$. Note that since all machines
of a type are identical, it does not matter which exact machines are
huge. Hence, the total number of combinations we need to enumerate
is bounded by $m^K$. For each type $\ell$, denote by \hugemach\ 
the huge and by \smallmach\ the non-huge machines of this type.
For each type $\ell$ we enumerate the $f(p,\epsilon)$ largest huge
jobs which are processed on a machine of type $\ell$, where $f(p,\epsilon)$
is a value which we obtain from the following proposition. We call
them the \emph{very huge} jobs and the corresponding machines the \emph{very huge machines}.
\begin{proposition}
\label{prop:f(p)}For each $p > 1$ and each $\epsilon>0$ there is a
number $f(p,\epsilon)$ such that $\sum_{g\in G} g^{p}+(2\cdot\min\{g\in G\})^{p}\le(1+\epsilon)^{p}\cdot\sum_{g\in G} g^{p}$
for any set $G$ of positive reals with $|G|\ge f(p,\epsilon)$. 
\end{proposition}
Knowing the very huge jobs for each type will imply later that we
can afford making certain mistakes when assigning the remaining huge
jobs. Note that there are at most $n^{K\cdot f(p,\epsilon)}$ possibilities
to enumerate. Also, for each type $\ell$ we guess the longest job
which is scheduled on a machine of type $\ell$ and which is not huge.
Denote by $c_{\max,\ell}$ its length. There are at most $n^{K}$
possibilities for this. 
\begin{lemma}
\label{lem:load-difference}
Consider an optimal solution and let $i$ be any machine of type $\ell$ that is not huge. Then its load is at least $c_{\max,\ell}$. Moreover, the loads of any two non-huge machines of type $\ell$ differ by at most $c_{\max,\ell}$. 
\end{lemma}
\begin{proof}
The claim follows from the strict convexity of the objective function and an exchange
argument, see Appendix~\ref{sec:appendix} for details. \qed
\end{proof}
For any machine $i$ of type $\ell$ we call a job $j$ 
\emph{huge on machine $i$} if $c_{i,j} > c_{\max,\ell}$. For each type
$\ell$ denote by $H_{\ell}$ the jobs that are huge on machines of
type $\ell$ and which are not longer than the $f(p,\epsilon)$ very
huge jobs for type $\ell$ which we guessed above. 
Finally, for each type $\ell$ we guess a value $\alpha_{\ell}\in\{1,...,n\}$
such that in the optimal solution the load of each machine of type
$\ell$ is at least $\alpha_{\ell}\cdot c_{\max,\ell}$ and at most
$(\alpha_{\ell}+2)\cdot c_{\max,\ell}$. Due to Lemma~\ref{lem:load-difference}
such a value $\alpha_{\ell}$ must exist. Note that there are at most
$n^{K}$ possibilities for $\alpha_{\ell}$. Next, we enumerate the
patterns of the big jobs on the non-huge machines of each type. To
this end, we define a job $j$ to be \emph{large} on a machine $i$
of type $\ell$ if $c_{i,j}>\epsilon\cdot\alpha_{\ell}\cdot c_{\max,\ell}$
and \emph{small} otherwise. Like in Section~\ref{sec:Makespan-minimization}
we enumerate over the (polynomial number of) patterns for each type.
From now on, assume that we know the correct values for all enumerated
quantities.

With this preparation, we formulate the remaining problem as a convex
program, which we denote as Slot-CP. Like in Section~\ref{sec:Makespan-minimization},
denote by $S$ the set of slots for the big jobs. If a job $j$ fits
into a slot $s$ then we introduce a variable $x_{s,j}$ which indicates
whether $j$ is assigned to $s$. For each combination of a job $j$
and a machine $i$ such that $j$ is small on $i$, we introduce a
variable $x_{i,j}$. Finally, if a job $j$ is huge on machines of
type $\ell$ (i.e., $c_{i,j}>c_{\max,\ell}$) then we introduce a variable
$x_{\ell,j}$ indicating whether $j$ is assigned to one of the huge
machines of type $\ell$. For each machine $i$, let $B_{i}$ denote
the total length of the large jobs on $i$. Let $\allsmallmach:=\cup_{\ell} \smallmach$
denote the set of all machines which are not huge. For each type $\ell$,
denote by $\topmach$ the very huge machines. For any very huge
machine $i\in \topmach$ we define a constant $t_{i}^{*}$ denoting
its load (due to its guessed job) and for any machine $i\in \allsmallmach$
we introduce a variable $t_{i}$ which models its load. We solve the following
convex program to an additive error of $\epsilon$. This can be done
in polynomial time since we have a separation oracle and the objective is convex and differentiable \cite{Groetschel:1988}. 
\newpage

\[
\textrm{(Slot-CP)}\quad\min\sum_{i\in \allsmallmach}(t_{i}+B_{i})^{p}+\sum_{\ell \in \types}\sum_{i\in \topmach}(t_{i}^{*})^{p}+\sum_{\ell \in \types}\sum_{j\in J}x_{\ell,j}\cdot(c_{\ell,j})^{p}
\]
\begin{align}
 &  & \sum_{i\in \allsmallmach}x_{i,j}+\sum_{s\in S}x_{s,j}+\sum_{\ell \in \types}x_{\ell,j} & =1 &  & \forall j\in J\\
 &  & \sum_{j\in H_{\ell}}x_{\ell,j} & \le h_{\ell} &  & \forall \ell \in \types \\
 &  & \sum_{j\in J}x_{s,j} & \le1 &  & \forall s\in S\nonumber \\
 &  & \sum_{j\in J}c_{i,j}\cdot x_{i,j} & \leq t_{i} &  & \forall i\in \allsmallmach \nonumber \\
 &  & \alpha_{\ell}\cdot c_{\max} & \le t_{i} &  & \forall \ell \in \types, \forall i\in \smallmach \\
 &  & x_{i,j} & \ge0 &  & \forall i\in \allsmallmach,\forall j\in J\nonumber \\
 &  & x_{s,j} & \ge0 &  & \forall s\in S,\forall j\in J\nonumber \\
 &  & x_{\ell,j} & \ge0 &  & \forall \ell \in \types, \forall j\in J \nonumber\\
 &  & t_{i} & \ge0 &  & \forall i\in \allsmallmach. \nonumber 
\end{align}
Since Slot-CP is a relaxation of the original problem, its optimal
value yields a lower bound on the optimum. Denote by $t_{i}^{*}$
the values obtained for the $t_{i}$-variables from the optimal solution
of (Slot-CP). In order to be able to use the concept of extreme point
solutions, we derive the following \emph{linear} program where all the
$t_{i}^{*}$'s are \emph{constants}.

\begin{align}
{\textrm{(Slot-LP)}} &  & \min\sum_{\ell \in \types}\sum_{j\in J} (c_{\ell,j})^{p} \cdot x_{\ell,j} \label{eq:Lp-total_assigned}\\
 &  & \sum_{i\in \allsmallmach}x_{i,j}+\sum_{s\in S}x_{s,j}+\sum_{\ell \in \types}x_{\ell,j} & =1 &  & \forall j\in J\\
 &  & \sum_{j\in H_{\ell}}x_{\ell,j} & \le h_{\ell} &  & \forall \ell \in \types \label{eq:huge-jobs}\\
 &  & \sum_{j\in J}x_{s,j} & \le1 &  & \forall s\in S\label{eq:Lp-slot-constraint}\\
 &  & \sum_{j\in J}c_{i,j}\cdot x_{i,j} & \leq t_{i}^{*} &  & \forall i\in \allsmallmach\label{eq:makespan-bound}\\
 &  & x_{i,j} & \ge0 &  & \forall i\in \allsmallmach,\forall j\in J\nonumber \\
 &  & x_{s,j} & \ge0 &  & \forall s\in S,\forall j\in J\nonumber \\
 &  & x_{\ell,j} & \ge0 &  & \forall \ell \in \types, \forall j\in J. \nonumber 
\end{align}
We devise an iterative rounding algorithm which computes an integral
solution whose overall value is only by a $(1+3\eps)^p$-factor bigger
than than the optimal value of (Slot-CP). Like in Section~\ref{sec:Makespan-minimization}
we work with linear programs $LP_{t}$ where $LP_{0}$ is the Slot-LP
and each $LP_{t}$ is obtained by taking $LP_{t-1}$ and fixing some
variables and removing some constraints. In each iteration $t$, we
compute an extreme point solution $x^{*}$ of $LP_{t}$.
\begin{lemma}
\label{lem:sparse-solutions-Lp}Let $x^{*}$ be an extreme point solution
for the linear program $LP_{t}$ for some iteration $t$. Then either 
\begin{enumerate}
\item there is a machine $i$ with at most two small jobs $j$ such that
$x_{i,j}^{*}\in(0,1)$, or
\item there is a slot $s$ with at most two jobs $j$ such that with
$x_{s,j}^{*}\in(0,1)$, or 
\item there is a type $\ell$ with at most two jobs $j\in H_{\ell}$ such
that $x_{\ell,j}^{*}\in(0,1)$.
\end{enumerate}
\end{lemma}
\begin{proof}
We can follow a similar argumentation as in Lemma~\ref{lem:extreme-point-solutions}, via the total number of non-integral variables.\qed
\end{proof}
First, we first fix all variables which are integral. If either case
1 or case 2 of Lemma~\ref{lem:sparse-solutions-Lp} applies we do
the same operation as in Section~\ref{sec:Makespan-minimization},
i.e., drop a constraint of type~\eqref{eq:makespan-bound} or replace
two jobs by an artificial job and drop a constraint of type~\eqref{eq:Lp-slot-constraint}.
If case 3 applies, i.e., if there is a type $\ell$ with at most two
jobs $j\in H_{\ell}$ such that $x_{\ell,j}^{*}\in(0,1)$, then we
define a schedule for the huge machines of type~$\ell$ by assigning
each integrally assigned huge job on a single (huge) machine and assign
the two fractionally assigned huge jobs together on one of the machines
of type $\ell$ which we defined to be huge. We call the latter machine
the \emph{improper} machine of type $\ell$. Then we remove the constraint~\eqref{eq:huge-jobs}
for type $\ell$.
We will show later that
the cost of the improper machine is very small in comparison with
the cost of the very huge machines of the respective type. After the
last iteration, we replace the introduced artificial jobs by the original
jobs that they subsumed (see Lemma~\ref{lem:replace-artificial-jobs}).

For each machine $i$, let $g_{i}$ denote its load in the computed
integral solution. With a similar reasoning as in Section~\ref{sec:Makespan-minimization}
we can show the following lemma.
\begin{lemma}
\label{lem:Lp-small-increase}For each small machine $i \in \smallmach$ it holds
that $g_{i}\le t_{i}^{*}+3\epsilon\cdot\alpha_{\ell}\cdot c_{\max,\ell}\le(1+3\eps)t_{i}^{*}$.
\end{lemma}
Apart from the improper machines, the cost of the huge machines does
not increase due to our rounding scheme, as the next lemma shows.
\begin{lemma}
\label{lem:Lp-huge-cost}Let $x^{*}$ be an optimal solution to the
Slot-LP and let $\ell$ be a type. The cost of its huge machines is
bounded by

\[
\sum_{i\in \hugemach} g_{i}^{p}\le\sum_{i\in \topmach}(t_{i}^{*})^{p}+(2\cdot\min\{t_{i}^{*}|i\in \topmach\})^{p}+\sum_{j \in H_\ell}x_{\ell,j}^{*}\cdot(c_{\ell,j})^{p}.
\]

\end{lemma}
Finally, we show that the cost of the improper machines is small in
comparison to the cost of the very huge machines. This follows from
Proposition~\ref{prop:f(p)} and Lemma~\ref{lem:Lp-huge-cost}.
\begin{lemma}
\label{lem:Lp-improper-machines}Let $\ell$ be a type. Then $\sum_{i\in \topmach}(t_{i}^{*})^{p}+(2\cdot\min\{t_{i}^{*}|i\in \topmach\})^{p}\le(1+\epsilon)^{p}\cdot\sum_{i\in \topmach}(t_{i}^{*})^{p}.$
\end{lemma}
Using Lemmas~\ref{lem:Lp-small-increase}, \ref{lem:Lp-huge-cost},
and \ref{lem:Lp-improper-machines} one can show that the total cost
of the final solution is at most by a factor $(1+3\eps)^p$ larger than
the total cost of the optimal solution of Slot-CP (see Lemma \ref{lem:combine} in the Appendix). This yields our
main theorem.
\begin{theorem}
Let $K\in\mathbb{N}$ and $p > 1$ be fixed. For any
$\epsilon>0$ there is a $(1+\epsilon)$-approximation algorithm for
assigning jobs to unrelated machines of $K$ types to minimize the
$L_{p}$-norm of the load vector. 
\end{theorem}

\bibliography{references}
\bibliographystyle{abbrv}

\newpage{}

\appendix

\section{Appendix}\label{sec:appendix}

\begin{proof}[Proposition \ref{prop:largify}]
By the pigeonhole principle there can be at most $D/\epsilon$ large jobs in $J_i$. Thus, 
$$\sum_{j\in \bigjobs}\max\{c_{i,j}^{d},\eps^2/D\}-\sum_{j\in \bigjobs}c_{i,j}^{d} \le (D/\eps) \cdot \eps^2/D = \eps$$
and thus $\sum_{j\in \smalljobs}c_{i,j}^{d}+\sum_{j\in \bigjobs}\max\{c_{i,j}^{d},\eps^2/D \} \le1+\epsilon$ for each $d$.
\qed
\end{proof}

\begin{proposition}\label{pro:remove-machine}
After removing a machine $i$ and the variable $x_{i,j}$ for each job $j$, the new linear
program $LP_{t+1}$ is feasible. \end{proposition}
\begin{proof}
Let $j_{1}$ be a removed job. We decrease $x_{i',j_{1}}^{*}$ and
$x_{s,j_{1}}^{*}$ to zero for any remaining machine $i'$ or slot
$s$. This cannot violate any residual constraint. \qed\end{proof}
\begin{proposition}\label{pro:remove-slot}
Let $s$ be a slot and suppose that in $LP_{t}$ for any job $j$ with $j\ne j_1$ and $j \ne j_2$
it holds that $x_{s,j}\in \{0,1\}$.
After replacing $j_{1}$ and $j_{2}$ by the artificial job $j_{0}$
and removing $s$, the new linear program $LP_{t+1}$ is feasible. \end{proposition}
\begin{proof}
If $x^{*}$ is feasible for $LP_{t}$, consider the solution $\hat{x}$
defined as: 
\begin{align*}
\hat{x}_{i,j} & :=\begin{cases}
x_{i,j}^{*}, & j\neq j_{0}\\
x_{i,j_{1}}^{*}+x_{i,j_{2}}^{*} & j=j_{0},
\end{cases}\\
\hat{x}_{s,j} & :=\begin{cases}
x_{s,j}^{*}, & j\neq j_{0}\\
x_{s,j_{1}}^{*}+x_{s,j_{2}}^{*}, & j=j_{0}.
\end{cases}
\end{align*}
 It is not hard to check that due to the construction, constraints
\eqref{eq:slot-constraint} and \eqref{eq:idle-constraint} of $LP_{t+1}$
will be satisfied by $\hat{x}$. For example, for constraint \eqref{eq:idle-constraint}
and the case that $j_{1}$ and $j_{2}$ are both small on machine
$i$, 
\begin{align*}
\sum_{j}c_{i,j}^{d}\hat{x}_{i,j} & =\sum_{j\neq j_{0}}c_{i,j}^{d}x_{i,j}^{*}+c_{i,j_{0}}^{d}(x_{j_{1}}^{*}+x_{j_{2}}^{*})\\
 & =\sum_{j\neq j_{0}}c_{i,j}^{d}x_{i,j}^{*}+c_{i,j_{1}}^{d}x_{j_{1}}^{*}+c_{i,j_{2}}^{d}x_{j_{2}}^{*}\\
 & \le\rema^{d}(i).
\end{align*}

\end{proof}
Constraint \eqref{eq:total_assigned} may be violated by $j_{0}$
in the sense that 
$\sum_{i}\hat{x}_{i,j_{0}}+\sum_{s}\hat{x}_{s,j_{0}}$ may be larger
than one; however, in that case it suffices to scale uniformly down
the values $(x_{i,j_0})_{i \in M}$, $(x_{s,j_0})_{s \in S}$ to obtain a feasible solution. \qed
\begin{figure}
\begin{center}
\begin{tikzpicture}[level distance=1.5cm,stealth-]%
\tikzstyle{ajob}=[circle,draw=black]
\tikzstyle{rjob}=[circle,draw=black,fill=gray]
\tikzstyle{slot}=[rectangle,draw=black,minimum height=.3cm,minimum width=.3cm]
\tikzstyle{level 1}=[sibling distance=3cm]
\tikzstyle{level 2}=[sibling distance=2cm]
\tikzstyle{level 3}=[sibling distance=1cm]
 \node[ajob,label={right:$j$}] {} [grow=left] 
    child {node[slot] {} 
      child {node[ajob] {} 
      	child {node[slot] {} 
			child {node[ajob] {} 
				child {node[slot] {} 
					child {node[rjob] {} 
					}
					child {node[rjob] {} 
					}
				}
			}
			child {node[rjob] {} 
			}
		}
      }
      child {node[ajob] {} 
      	child {node[slot] {} 
			child {node[rjob] {} 
			}
			child {node[rjob] {} 
			}
		}
      }
    }
;
\end{tikzpicture}
\end{center}

\caption{\label{fig:tree}Example of the construction in Lemma \ref{lem:tree}.
Gray circles represent original jobs, white circles represent artificial
jobs, and squares represent slots.}
\end{figure}

\begin{proof}[of Lemma~\ref{lem:tree}]
To obtain the sets $J_{j}$ and $S_{j}$ we trace back, in the execution
of the algorithm, which jobs have been subsumed by $j$, which jobs
have been subsumed by those jobs, and so on. Associate to the
execution of the rounding algorithm a bipartite graph $G$ whose nodes
are the jobs and slots that have been removed by the algorithm. Whenever
two jobs $j_{1}$, $j_{2}$ competing for slot $s$ are subsumed by
an artificial job $j_{0}$, insert the arcs $(j_{1},s)$, $(j_{2},s)$,
$(s,j_{0})$ in $G$. The resulting graph will be a directed forest
in which every original job is a leaf, every slot has indegree two,
and every artificial job has indegree one (see Figure~\ref{fig:tree}
for an example). We define $J_{j}$ to be the leaves of the tree of
$G$ having $j$ as the root, while $S_{j}$ is the set of slot nodes
of the same tree. Properties (1)--(3) are now clear from the definition
of the algorithm. Properties (4) and (5) follow from the fact that
the slots in $S_{j}$ and jobs in $J_{j}$ have all been removed from
the linear program. Property (6) (reminiscent of the fact that, in
a tree in which every internal node has two children, the number of
internal nodes is equal to the number of leaves, minus one) is proved
by induction on the structure of the tree: consider a slot $s$ having
two jobs $j_{1},j_{2}\in J_{j}$ as children, find inductively a feasible
assignment for the smaller tree where $j_{1}$, $j_{2}$ and $s$
have been removed and the artificial job that is the parent of $s$
has been replaced with whichever of $j_{1},j_{2}$ fits $j_{0}$'s
parent, and assign the other job to $s$. \qed\end{proof}

\begin{proof}[of Proposition~\ref{prop:f(p)}]
Take $f(p,\eps) = (1+2^p)/(1+\eps)^p$. Then 
\begin{align*}
\left( \sum_{g \in G} g^p + (2 \min_{g \in G} g)^p \right)^{1/p} 
&\le \left( \norm[p]{g}^p  + 2^p \frac{\norm[p]{g}^p}{|G|} \right)^{1/p} \\
&= \left(\frac{1+2^p}{|G|}\right)^{1/p} \norm[p]{g} \\
&\le (1+\eps) \norm[p]{g}. 
\end{align*}
Hence, raising both sides to the power of $p$, we obtain the claim. \qed
\end{proof}

\begin{proof}[of Lemma~\ref{lem:load-difference}]
The claim will follow from the convexity of the objective function and an exchange
argument. 
Consider the first part of the claim; the second part is proved similarly. 
Let $i'$ be any machine of type $\ell$ where a job of length $c_{\max,\ell}$ is assigned, and let $a_0$, $a_1$ be the loads of machines $i$, $i'$ respectively; by construction $a_1 > c_{\max,\ell}$. Assume by contradiction that $a_0 < c_{\max,\ell}$. Consider the assignment where we exchange all the jobs on $i$ with the job on $i'$ of length $c_{\max,\ell}$. Since the machines are of the same type, the processing times of the jobs are unaffected. Call $a_0'$, $a_1'$ the new loads of $i$,$i'$. Observe that $a_0 < a_0' < a_1$ and $a_0 < a_1' < a_1$; equivalently, there exist $\mu,\eta \in (0,1)$ such that $a_0' = \mu a_0 + (1-\mu) a_1$, $a_1' = \eta a_1 + (1-\eta) a_0$. However, since $a_0'+a_1' = a_0+a_1$ and $a_0<a_1$, one has $\mu=\eta$. Then $(a_0',a_1') = \mu \cdot (a_0,a_1) + (1-\mu) \cdot (a_1,a_0)$. Consequently, if $a$ and $a'$ are the load vectors of all machines before and after the exchange, and $a''$ is the load vector obtained from $a$ by exchanging the $i$-th and $i'$-th component, 
$ a' = \mu a + (1-\mu) a''. $ 
Finally, from the strict convexity of the $L_p$ norm (for $1 < p <\infty$), 
$$\norm[p]{a'} = \norm[p]{ \mu a + (1-\mu) a'' } < \mu \norm[p]{a} + (1-\mu) \norm[p]{a''} = \norm[p]{a}, $$
which contradicts the assumed optimality of the solution. \qed
\end{proof}

\begin{lemma}
\label{lem:objective-not-increase}Let $T'$ be the set of all types 
$\ell$ for which there is still the constraint $\sum_{j\in H_{\ell}}x_{\ell,j}\le h_{\ell}$
in $LP_{q}$ and $LP_{q+1}$. Then 
\[
\sum_{\ell\in T'}\sum_{j\in H_{\ell}}x_{\ell,j}^{q}\cdot(c_{\ell,j})^{p}\ge\sum_{\ell\in T'}\sum_{j\in H_{\ell}}x_{\ell,j}^{q+1}\cdot(c_{\ell,j})^{p}
\]
 where $x^{q}$ and $x^{q+1}$ denote optimal solutions for $LP_{q}$
and $LP_{q+1}$, respectively.
\end{lemma}
\begin{proof}
According to the proof of Proposition~\ref{pro:remove-machine}, when in our iterative rounding routine we drop a constraint of type~\eqref{eq:makespan-bound} the optimal solution for $LP_{q}$ yields a feasible solution for $LP_{q+1}$. In particular, the value of the optimal solution does not increase.
By the proof of Proposition~\ref{pro:remove-slot} the same holds when we replace two jobs by an artificial job. Also, when we assign jobs to the huge machines of some type~$\ell$ then the remaining variable
assignment stays feasible. Hence, the optimal objective value for $LP_{q+1}$ is upper bounded by the optimal objective value for~$LP_{q}$. This implies the claim of the lemma. \qed
\end{proof}


\begin{proof}[of Lemma~\ref{lem:Lp-huge-cost}]
For the very huge jobs the contribution is clear. For the huge jobs
excluding the very huge jobs using Lemma~\ref{lem:objective-not-increase}
and the definition of our operation for the removing the constraints
from Inequality~\eqref{eq:huge-jobs} we get a bound of $(2\cdot\min\{t_{i}^{*}|i\in \topmach\})^{p}+\sum_{j}x_{\ell,j}^{*}\cdot(c_{\ell,j})^{p}$. \qed
\end{proof}

\begin{lemma}
\label{lem:combine}
For the loads $g_{i}$ of the computed solution it holds that
\[
\norm[p]{g}^{p}\le(1+3\epsilon)^p \, \mathrm{OPT}_{\mathrm{CP}} \le(1+3\epsilon)^p\, \mathrm{OPT}^p
\]
where $\mathrm{OPT}_{\mathrm{CP}}$ denotes the value of an optimal solution
of Slot-CP and $\mathrm{OPT}^p$ denotes the value of an optimal integral solution of Slot-CP.
\end{lemma}
\begin{proof}
Recall that the vector $t^{*}$ indicates the loads of the very
huge machines due to the guessed very huge jobs and the loads
of the small machines due to the allocation of the small jobs by Slot-CP.
Then

\begin{eqnarray*}
\norm[p]{g}^p & \le & 
\sum_{\ell \in \types}\sum_{i\in \smallmach}(t_{i}^{*}+3\epsilon\cdot\alpha_{\ell}\cdot c_{\max,\ell})^{p}+\sum_{\ell \in \types}(1+\epsilon)^{p}\sum_{i\in \topmach}(g_{i})^{p}+\sum_{\ell \in \types}\sum_{j}x_{\ell,j}^{*}\cdot(c_{\ell,j})^{p}
\\
 & \le & 
 \sum_{\ell \in \types}\sum_{i\in \smallmach}(t_{i}^{*}+3\epsilon\cdot t_{i}^{*})^{p}+\sum_{\ell \in \types}(1+\epsilon)^{p}\sum_{i\in \topmach}(g_{i})^{p}+\sum_{\ell \in \types}\sum_{j}x_{\ell,j}^{*}\cdot(c_{\ell,j})^{p}\\
 & \le & (1+3\epsilon)^p \left(
 \sum_{\ell \in \types}\sum_{i\in \smallmach}t_{i}^{*}+\sum_{\ell \in \types}\sum_{i\in \topmach}(t_{i}^{*})^{p}+\sum_{\ell \in \types}\sum_{j}x_{\ell,j}^{*}\cdot(c_{\ell,j})^{p} \right) 
 \\ 
 & \le & (1+3\epsilon)^p \, \mathrm{OPT}_{\mathrm{CP}}\\
 & \le & (1+3\epsilon)^p \, \mathrm{OPT}^p. 
\end{eqnarray*}
\qed
\end{proof}

\end{document}